\newif\ifreview

\documentclass[sigconf]{acmart}

\usepackage[T1]{fontenc}
\usepackage[fleqn]{mathtools}
\usepackage{tabularx}
\usepackage{float}
\usepackage{xcolor}
\usepackage{listings}
\usepackage{fancyvrb}
\usepackage{pifont} 

\AtBeginDocument{%
  }

\copyrightyear{2026}
\acmYear{2026}
\setcopyright{cc}
\setcctype{by}
\acmConference[UIST '26]{The 39th Annual ACM Symposium on User Interface Software and Technology}{November 02--05, 2026}{Detroit, MI, USA}
\acmBooktitle{The 39th Annual ACM Symposium on User Interface Software and Technology (UIST '26), November 02--05, 2026, Detroit, MI, USA}
\acmDOI{10.1145/3830398.3830614}
\acmISBN{979-8-4007-2856-3/2026/11}

\begin{document}

\title[Explorable Proofs]{Explorable Proofs: Making Written Proofs Explorable by Grounding Them in Formal Representations}


\definecolor{andrewpurple}{HTML}{A53DFF}
\definecolor{andreworange}{HTML}{E07400}
\definecolor{darkgreen}{HTML}{009B55}
\definecolor{darkblue}{HTML}{004d80}
\definecolor{magenta}{HTML}{99195d}

\newcommand\andrew[1]{\textcolor{andrewpurple}{#1}}
\newcommand\important[1]{\textcolor{darkgreen}{#1}}
\newcommand\unimportant[1]{\textcolor{gray}{\sout{#1}}}
\newcommand\move[1]{\textcolor{andreworange}{#1}}
\newcommand{\change}[1]{#1}
\newenvironment{changes}
{\begingroup}
{\endgroup}

\def\computer#1{{\small\texttt{#1}}}
\AtBeginEnvironment{quote}{\itshape}

\def\subparagraph#1{\textbf{#1.}}

\def\UrlBigBreaks{\do\/\do-\do\#}

\def\shortspace{\kern 0.1em}

\def\KaTeX{K\kern-.2em\raisebox{.2em}{\scriptsize A}\kern-.12em\TeX}

\definecolor{niceblue}{HTML}{8295ff}
\def\bigbox{\color{niceblue}\rule[.25ex]{1ex}{1ex} \hskip .1ex}
\def\smallbox{\hskip .25ex \color{gray}\rule[.5ex]{.5ex}{.5ex} \hskip .25ex \hskip .1ex}
\def\boxes#1#2{
\hskip .1ex 
\newcount\boxnum
\boxnum=0
\loop
\ifnum \boxnum<#1 \bigbox \else \smallbox \fi

\advance \boxnum by 1
\ifnum \boxnum<#2
\repeat
}

\newenvironment{inlinefigureenv}
{\setlength{\topsep}{2.5ex}\center}
{\endcenter}

\newcommand{\inlinefigure}[2][.5\textwidth]{%
\begin{inlinefigureenv}%
\includegraphics[width=#1]{#2}%
\vspace{-1.25ex}%
\end{inlinefigureenv}%
}

\begin{abstract}

LLM-generated explanations can make technical content more accessible, but there is a ceiling on what they can support interactively. Because LLM outputs are static text, they cannot be executed, tinkered with, or stepped through. \change{In this paper, we explore} how grounding explanations in a formalized representation enables interactive affordances beyond what static text supports. We instantiate this idea for mathematical proof comprehension in \textit{explorable proofs}, \change{an approach} that uses LLMs to translate a theorem and its written proof into Lean, a programming language for machine-checked proofs, and links the written proof with the Lean code. Readers can work through the proof at a step-level granularity, test custom examples or counterexamples, and trace the logical dependencies bridging between steps. Each worked-out step is produced by executing the Lean proof on an example and extracting its intermediate state. A user study ($n = 16$) shows potential advantages of this approach: \change{in a proof-understanding task, participants who used the explorability features wrote proof descriptions that an instructor preferred over those written with a proof-and-chatbot baseline, and gave more correct and more detailed answers to proof comprehension questions.}


\end{abstract}

\author{Hita Kambhamettu}
\orcid{0000-0001-9620-1533}
\email{hitakam@seas.upenn.edu}
\affiliation{%
  \institution{University of Pennsylvania}
  \city{Philadelphia, PA}
  \country{USA}
}

\author{Will Crichton}
\orcid{0000-0001-8639-6541}
\email{will_crichton@brown.edu}
\affiliation{%
  \institution{Brown University}
  \city{Providence, RI}
  \country{USA}
}

\author{Sean Welleck}
\orcid{0009-0003-1744-0096}
\email{swelleck@andrew.cmu.edu}
\affiliation{%
  \institution{Carnegie Mellon University}
  \city{Pittsburgh, PA}
  \country{USA}
}

\author{Harrison Goldstein}
\orcid{0000-0001-9631-1169}
\email{hgoldste@buffalo.edu}
\affiliation{%
  \institution{University at Buffalo, SUNY}
  \city{Buffalo, NY}
  \country{USA}
}

\author{Andrew Head}
\orcid{0000-0002-1523-3347}
\email{head@seas.upenn.edu}
\affiliation{%
  \institution{University of Pennsylvania}
  \city{Philadelphia, PA}
  \country{USA}
}

\renewcommand{\shortauthors}{Kambhamettu et al.}


\newcommand{\circlenum}[1]{%
  \textcolor[HTML]{7f7f7f}{%
    \ding{\numexpr 201 + #1\relax}%
  }%
}

\begin{teaserfigure}
  \centering
  \includegraphics[width=0.95\textwidth]{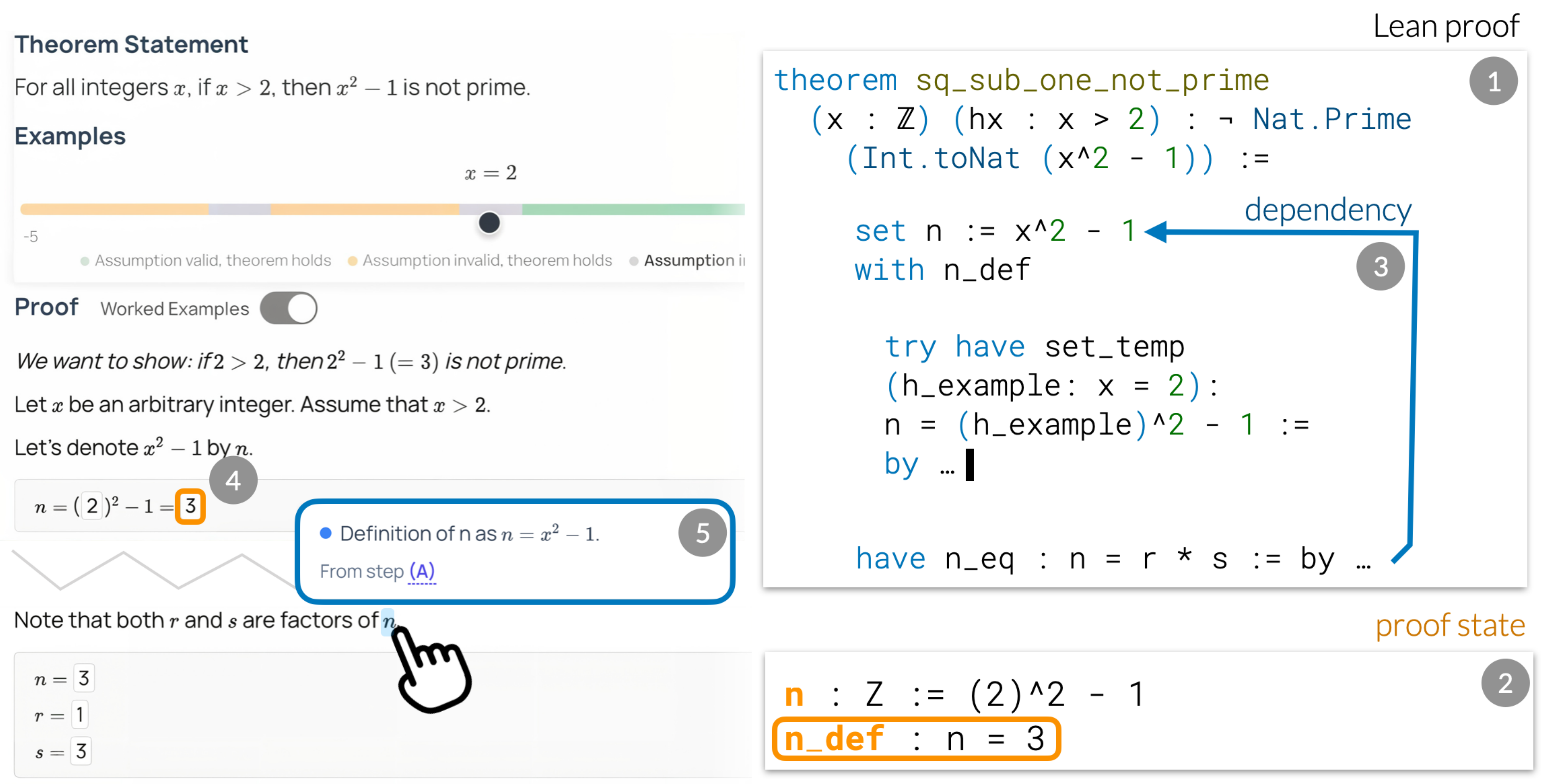}
  \vspace{-1ex}
  \caption{
  Explorable proofs are written proofs made explorable by analyzing and executing a linked machine proof. \textmd{Given a written proof, a corresponding machine-checkable Lean proof is generated~(\circlenum{1}). Structural and semantic information is derived by executing steps of the Lean proof ~(\circlenum{2}) and extracting logical dependencies~(\circlenum{3}). This information is plugged into the written proof, allowing readers to use it as an explorable text. They can ask questions about whether a theorem holds on one or another example, the bounds of applicability, how the math ``works out'' for individual examples~(\circlenum{4}), and where counterexamples fail. They can also ask questions about why propositions matter and where they come from by tracing to their definitions and uses~(\circlenum{5}).} 
  }
  \vspace{2ex}
  \Description{A two-panel screenshot. The left panel shows the explorable proofs reading interface for the theorem "For all integers x, if x is greater than 2, then x squared minus 1 is not prime." At the top is the theorem statement, below it a horizontal slider set to x equals 2, with the track colored in segments indicating where the assumption holds and where the conclusion holds. Below the slider is the written proof, with a toggle labeled "Worked Examples" switched on. Each proof step has a boxed worked example beneath it showing the step evaluated at x equals 2, such as "n equals 2 squared minus 1 equals 3." A cursor hovers over the variable n in one step, and a popup shows the definition of n and a link back to the step where it was introduced. The right panel shows the corresponding Lean proof as code, beginning with "theorem sq_sub_one_not_prime," with a blue arrow labeled "dependency" pointing from a later line back to the line that defines n. Below the code is a box labeled "proof state" containing the hypothesis n_def with value n equals 3. Five numbered circular markers link the two panels: the Lean proof, the proof state, the dependency arrow, the worked-out step, and the definition popup.
}
  \label{fig:teaser}
\end{teaserfigure}


\maketitle

\section{Introduction}

Artifacts that express formal reasoning, like mathematical proofs~\cite{selden2003validations, schoenfeld1988good}, programs~\cite{lister2004multi, letovsky1987cognitive}, and logical arguments~\cite{mercier2011humans, johnson2005mental}, can be difficult to understand by simply reading them. To illustrate this challenge, consider the following theorem:

\begin{quote}
Let $a,d \in \mathbb{N}$ with $\gcd(a,d)=1$. Then there are infinitely many prime numbers $p$ such that
\[
p \equiv a \pmod d.
\]
\end{quote}

\noindent
It says: pick any two numbers $a$ and $d$ that share no common factors; then infinitely many primes will leave remainder $a$ when divided by $d$. To get a feel for the theorem, a reader might instantiate it with $a=1$, $d=4$, where $\gcd(1,4)=1$ holds, and check which primes have remainder $1$ when divided by $4$, seeing enough primes that they get a sense that they might go on forever:
\[
5, 13, 17, 29, 37, \dots \equiv 1 \pmod 4.
\]

The reader might then wonder why $a$ and $d$ must be coprime. At $a=2$, $d=4$, where $\gcd(2,4)\neq 1$, only $p=2$ satisfies $p \equiv 2 \pmod 4$. All other numbers in this progression ($6$, $10$, $14$...) are multiples of $2$ and thus composite. This example clarifies that if $a$ and $d$ share a common factor, all numbers in the sequence will share that factor, limiting the sequence to at most one prime.\footnote{\change{This is an aspirational example: this proof is out of scope for our implementation, though the illustrated pattern of inquiry aligns with what our interface supports.}} 

Through this exercise, the reader is actively exploring the theorem and its bounds: choosing inputs, running the theorem against them, and using the results to build a working model of what the theorem's bounds are, and why they exist. This kind of active engagement, illustrating abstract assertions through concrete examples, is a well-documented strategy for effective proof reading~\cite{weber2015effective}.


AI tools are increasingly used for math education~\cite{pepin2025scoping, mohammad2025college}. While some students seem to be putting the tools to productive use~\cite{kock2025engineering, mohammad2025college}, conventional AI-generated answers can insert distance between readers and the mathematical argument they are trying to understand. One can, of course, explore a proof by repeatedly prompting an AI chatbot for clarifications, but each answer is a new generation of freeform prose rather than a persistent object the reader can execute or manipulate. \change{The reader loses the benefits of rapid local feedback, instead reviewing disconnected copies of generated outputs.} In this work we ask: how can we combine the convenience of AI explanations with the exploration necessary for understanding formal statements in domains like math?

In this paper, we explore how grounding written math in a formal representation could enable interaction that plain text cannot support. We build on the idea of explorable explanations~\cite{victor2011explorable}, where digital documents augment narrative with computational models that let readers manipulate assumptions and observe consequences.

Recent HCI research has sought to deepen reading by deriving representations that enrich how readers interact with texts. For example, AbstractExplorer~\cite{gu2025abstractexplorer} infers schemas for aligning texts in a way that facilitates comparison. DataTales~\cite{sultanum2023datatales} derives connections between text and visualization to help readers cross-reference the two. VeriPlan~\cite{lee2025veriplan} casts text plans into temporal logic so that their satisfiability can be checked with model checkers. \change{These systems show how a derived representation can enable useful new textual interactions. Could such an approach be applied in a way that brings explorability to artifacts like mathematical proofs?}

In this work, we develop \textit{explorable proofs}, a \change{proof-of-concept approach} to making written argument explorable by grounding it in executable formal representation. In our approach, we take a natural-language theorem and proof as input and generate a corresponding theorem and proof in the Lean interactive theorem proving language~\cite{demoura2015lean}. We use Lean to execute the proof on example inputs and extract intermediate proof states. These states give a formal record of what is known and what remains to be shown at each step of the proof, and they can be analyzed to find which propositions depend on which others. We use the formalization to enable exploration: a reader can choose an example value, step through the proof as it applies to that example, and trace which earlier facts each step draws on. \change{Our current prototype is designed for undergraduates in introductory proof courses, and our approach supports the proofs they commonly encounter: algebraically tractable theorems with numeric inputs, where forward-reasoning Lean proofs can structurally mirror their written counterparts.}




We evaluate explorable proofs in a user study ($n=16$) against a chatbot baseline, an increasingly common way readers seek help with proofs~\cite{yoon2024students}. \change{Participants using explorable proofs produced more correct and more detailed answers to proof comprehension questions. An instructor consistently ranked their proof descriptions higher than those written with the baseline.}

We conclude by discussing how our approach to grounding written explanations in formal representations may apply beyond proofs: formally-grounded content could become inspectable artifacts that readers can probe directly and use to review concrete instances, trace logical dependencies, and investigate boundary conditions.
In summary, this paper contributes:
\begin{itemize}
    \item \change{A blueprint for grounding explanations in AI-generated formal representations to enable affordances beyond what freeform text provides.}
    \item \textit{Explorable proofs}, an instantiation of this idea that enables example-driven exploration of proofs by using generated Lean as formal backing.
    \item Results from a user study ($n=16$) in which participants who used explorable proofs \change{wrote proof descriptions an instructor preferred over those written with a chatbot baseline and gave more correct and more detailed answers to comprehension questions.}
\end{itemize}

\section{Background and related work}


\subsection{$\forall x \in \{\text{theorem},\text{proof}\},\; \text{hard}(x)$}

Proof comprehension is challenging~\cite{stylianou2015undergraduate, harel2007toward, stylianou2009teaching}. Effective proof reading involves active reconstruction, where a proof reader works through examples and connects individual steps~\cite{weber2015effective}. Such habits are exemplified by expert proof readers, who look for the big picture \cite{weber2008how, mejiaramos2012assessment} and reusable proof methods \cite{weber2011why}, and test assertions with examples \cite{weber2011why, weber2008how}. Novices, in contrast, are prone to reason about proofs in the abstract and fail to use concrete examples that could help them develop intuition for theorems and their proofs \cite{moore1994making}.

While AI is increasingly used in math education~\cite{Serhan2024IntegratingCIA, Wardat2023ChatGPTARA}, it may hinder the kinds of processes useful to understanding proof. \citet{yoon2024students} point out the hazard of students prematurely relying on AI as authority when consulting a proof. Our work aims to equip readers with explanations that support the kinds of active exploration that lead to proof comprehension.

\subsection{$x \mapsto \text{example}(x)$}

Prior work has contributed a variety of ways to improve mathematical explanation. One approach has been to explain its notation better~\cite{head2022math} with improved utilities to style and label typeset math~\cite{wu2023ffl, tao2025freeform}. 
Another approach, taken by this work, is to make mathematical documents interactive. \citet{dragunov2003designing} described an aspirational vision for dynamic math interfaces. Among the proposed features were unrolling implicit steps of derivations and showing the provenance of expressions. Dragunov and Herlocker, however, did not build these features. Other work has. For instance, prior work has developed reactive documents with scrubbable variables~\cite{victor2011tangle}, gesture-based algebraic manipulation tools~\cite{weitnauer2016graspable,ref:zeleznik2010hands}, rich formula diagramming and glosses~\cite{crichton2021nota,li2022heartdown}, and formula markup that is intrinsically executable~\cite{li2022heartdown}. Authors of math texts and videos have experimented with yet more ways of presenting math~\cite{head2022math}. For instance, ~\citet{rougeux2018byrne}'s adaptation of Byrne's Euclid connects variables in written proof steps to relevant elements of corresponding visual proofs. Yet, producing such explanations can be very labor-intensive. 

\citeauthor{alcock2011eproofs}'s ``e-Proofs'' are perhaps the closest to our work in its intent to clarify proof through digital augmentation~\cite{alcock2011eproofs}. Among its features, e-Proofs make hidden logical structures explicit by using visual annotations to link individual statements to their premises. However, e-Proofs were videos rather than explorable artifacts. They also required manual authoring. Our system differs in these two aspects. While ``e-Proofs'' were well-received by students, they did not lead to measured improvements in comprehension~\cite{alcock2015investigating}, underscoring that augmentations to math do not necessarily bring about improvements in learning.

This paper is particularly interested in enabling example-based exploration.
In this goal, it aligns with a position put forward by Bret Victor in his ``explorable explanations'' \cite{victor2011explorable} and ``Kill Math'' \cite{victor2011killmath} projects. These projects emphasize the potential of well-designed interactive explanations to improve mathematical understanding over conventional static, symbolic report formats. Victor advocates for bypassing conventional notation entirely. Our approach, in contrast, builds a bridge between the symbolic and concrete, connecting executable math back to static proof steps. This is conceptually similar to ``concreteness fading,'' an instructional technique that has been used in mathematics to ease students from concrete manipulatives to more formal mathematical representations~\cite{kim2020concreteness}. More generally, the value of explorable interfaces in learning has been documented in initiatives like the PhET~\cite{phet} scientific simulators.

\change{Augmented Math and Augmented Physics~\cite{augmentedmath, augmentedphysics} apply the explorable explanation model to math and physics textbooks. These systems scaffold the reading of math and physics exposition and in some ways unpack math argument. For instance, they allow one to expand formulas and view how graphical depictions change as function parameters are altered. Our work builds on this by allowing deeper exploration of proof. We let users follow logical flow across steps, test theorems on specific values, and understand where the argument breaks down when assumptions are not met.}


\subsection{$\text{LLMs} \cap \text{Lean}$}

A formal representation of mathematical logic helps us generate explorable explanations for explorable proofs. The formal representation we use is Lean \cite{demoura2015lean}. Lean is an interactive theorem prover and functional programming language. It allows users to write proofs by programming. Users write \emph{tactics} that incrementally build up a machine-checkable argument that a particular mathematical goal is true. Lean belongs to a broader family of ``proof assistants'' or ``interactive theorem provers,'' such as Rocq and Isabelle~\cite{harrison2014history}.

Proof assistants can be useful for mathematics education because they provide immediate feedback and demand that students present their arguments unambiguously. That said, they can also be daunting to those who are not domain experts, in part due to their highly specialized syntax, rigidity, and the verbosity of full formalization~\cite{Minh2025ProofAF}. Tools like Verbose Lean~\cite{massot2024teaching} and Waterproof~\cite{wemmenhove2024waterproof} have sought to enable the advantages of proof assistants while mitigating some challenges, letting students write their proofs in a proof assistant where surface syntax more closely resembles paper proof.

\change{Production-grade proof assistants offer numerous features for helping understand proofs. One such feature is the VS Code Infoview panel, which lets Lean users inspect the proof state at any point in a proof file.} There are several efforts to make formalized proofs more communicative, such as Alectryon, which supports the publication of literate proofs~\cite{pit2020untangling} and the Lean widgets~\cite{LeanWidgets} framework for creating custom visualizations of  proof state information. We explore how to bring analogous features into augmented views of written mathematical proof.

\change{Our technical approach} uses LLMs to translate theorems and their written proofs into Lean. While recent advances show that LLMs are increasingly capable of generating valid Lean code \cite{jiang2023draft, liu2026numina}, the machine learning community has often treated formalization as a benchmark for automated reasoning, optimizing for whether a proof compiles, not whether it can be used to help understand written argument. Designing pipelines that formalize for human consumption remains largely underexplored. \change{Our work proposes a pipeline that generates formalizations structured for readers rather than only for the proof checker. }

\section{\change{Interface design}}

Explorable proofs are written mathematical proofs made explorable by grounding them in a formal representation. While we instantiate our idea for algebraically tractable proofs with numeric inputs, we believe the basic approach---pairing prose with an executable formal counterpart---could be used to build explanation interfaces in other formal domains (see Section~\ref{sec:future-work}). This section  describes the high-level design of the explorable proofs interface. \change{Appendix~\ref{sec:narrative-scenario} offers a scenario-centric introduction to the system.}

\change{Explorable proofs aim to offload what we consider incidental difficulty, like computing and tracking intermediate values across examples, to the interface while preserving the essential difficulty of understanding why each step holds and why it is necessary~\cite{bjork2011making} (see Section ~\ref{sec:incidental-diff}).} Explorable proofs provide the following affordances to help readers engage with proof.

\paragraph{Checking bounds of a theorem} When first reading a theorem, a reader might want to probe its limits: for what values does it hold, and for what values does it not? The \textit{theorem input slider} (Figure~\ref{fig:implementation}.1) lets readers evaluate the theorem on any input, giving them a sense of what the theorem is stating. Under the hood, the system evaluates the Lean proof on the chosen input and surfaces whether the conclusion holds at that input value.

\paragraph{Following steps of a proof} Once a reader understands what the theorem claims, they may want to read its proof. For each step of the written proof, the system displays a \textit{worked example} (Figure~\ref{fig:implementation}.2) that instantiates the step's argument with concrete values, helping the reader see which assertions hold under various assumptions and why. To enable this feature, the system extracts intermediate values from the Lean proof state on example values for the theorem's parameters chosen from the slider, or for any example values provided by the user in a custom value input widget. The system then inserts the intermediate values into generated string templates that show the worked-out math implied by each step.

\paragraph{Tracing provenance of a proposition} Proofs can be difficult to follow, especially when a step references facts established much earlier. It can also be unclear, on a first read, why a particular step is necessary. The \textit{proof step dependencies} feature (Figure~\ref{fig:implementation}.3) lets a reader click a proposition in a proof step to see which earlier steps it depends on and which later steps build on it, making the logical structure of the proof navigable. This is enabled by constructing a dependency graph over propositions in the Lean proof and mapping it back to the corresponding propositions in the written proof.

\paragraph{Understanding when a proof breaks} It can be difficult to understand why a theorem is stated the way it is, particularly why a particular assumption is necessary. By sliding to an input that violates the theorem's assumptions, a reader can \textit{inspect the proof logic} (Figure~\ref{fig:implementation}.4) and see exactly which step breaks down, making the role of each assumption legible. Under the hood, the system evaluates the Lean proof state on the given input and flags steps at which the proof no longer holds.


\begin{figure}
    \centering
    \includegraphics[width=\linewidth]{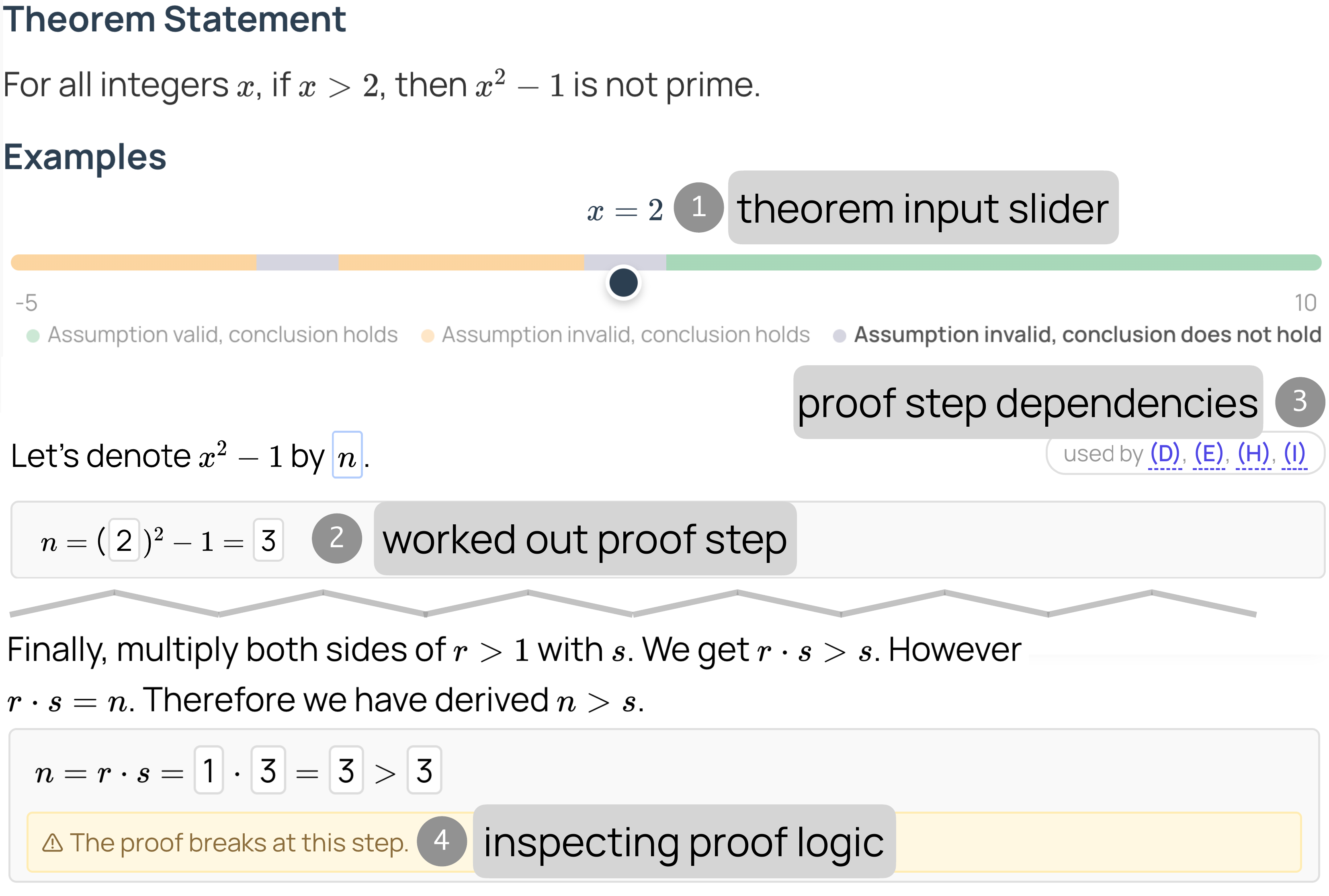}
    \caption{The explorable proofs interface for the theorem ``For all integers~$x$, if $x > 2$, then $x^2 - 1$ is not prime''. \mdseries The theorem input slider lets readers set a value of~$x$ (\circlenum{1}); color coding indicates \change{whether the theorem's conclusion holds for each value.} Each proof step is instantiated with Lean-verified intermediate values (\circlenum{2}): for $x = 2$, the expression ``$n = x^2 - 1$'' evaluates to~$3$. Clicking a variable such as~$n$ reveals the downstream proof steps that depend on it (\circlenum{3}). When a value is chosen that violates the theorem's assumptions, the interface flags exactly which assertions in the proof fail (\circlenum{4}). 
    \Description{A screenshot of the explorable proofs interface. At the top is the theorem statement ``For all integers x, if x > 2, then x squared minus 1 is not prime,'' followed by a horizontal slider set to x = 2 whose track is color-coded to show where the assumption and conclusion hold. Below is the written proof, each step followed by a boxed worked example evaluated at the selected input, such as n = 2 squared minus 1 = 3. One step is expanded to show which later steps depend on the variable n. A final step is flagged with a warning icon indicating the proof breaks at that step for the chosen input. Four numbered callouts label the theorem input slider, the worked-out proof step, the proof step dependencies, and the proof-logic inspection.}
    }
    \label{fig:implementation}
\end{figure}

\section{\change{Technical approach}}\label{System}

\change{Explorable proofs rely on a proof-of-concept backend to transform written proof into an augmented proof.}
Our high-level technical approach is to establish a tight coupling between a written proof and its formal Lean representation. Interactions such as checking bounds of a theorem, following steps of a proof, or tracing provenance of a proposition depend on data from a formalized Lean theorem and proof. Our approach follows four stages: (1) \textit{aligned generation}, in which a Lean proof is generated that is structurally similar to the written proof; (2) \textit{link-making}, in which an LLM maps Lean tactics and variables to their corresponding written steps and symbols; (3) \textit{dependency inference}, in which proof state diffs expose which propositions each tactic introduces and consumes; and (4) \textit{execution}, in which the proof is run on concrete inputs to extract intermediate values. Each stage involves implementation decisions to tighten the coupling between prose and formal proof. \change{We use Claude Code (Opus 4.6) for generating code (the Lean proof) and GPT-5.1 for generating natural language (link-making, worked-example templates). These choices were made based on how well each model performed on its stage of the pipeline.}

\newcommand{\code}[1]{%
  \begingroup%
  \setlength{\fboxsep}{1pt}%
  \colorbox{black!8}{\texttt{#1}}%
  \endgroup%
}

\paragraph{Aligned generation of Lean proof.}

Given a theorem and its written proof, we use a coding agent, Claude Code (Opus 4.6), to generate a corresponding Lean proof (Figure~\ref{fig:pipeline}.1). In Lean, a proof is written as a sequence of \emph{tactics}, instructions that are applied to the current proof state, the stated hypotheses and goals, to simplify and eventually prove the goals. To attempt to align the generated Lean proof to the original one, we instruct the agent to follow three rules. First, it should build up propositions in the same order as the prose. Second, the proof must use \code{have} statements, which introduce named, inspectable intermediate propositions into the proof. Third, each block of \code{have} statements should correspond to a step in the written proof. We instruct the LLM to annotate each block with a comment identifying its corresponding prose step and to use variable names that match variables in the written proof.

\paragraph{Link-making}
\label{sneaky-link}
We pass the Lean and written proofs to an LLM (GPT-5.1), which outputs a mapping from Lean proof steps to written proof steps. The LLM extracts two levels of links: from Lean blocks to written proof steps (using annotations generated in the last stage) and from written propositions to specific Lean variables. We use these links in later stages to channel information from the Lean proof into an augmented written proof (Figure~\ref{fig:pipeline}.6).




\paragraph{Dependency inference.}

Then, we recover the proof's dependency structure (Figure~\ref{fig:pipeline}.5). We instruct Claude Code to extract the proof state at each tactic step using \code{lean\_goal}~\cite{lean_lsp_mcp}. A proof state is a list of active hypotheses (named propositions currently in scope) together with the current goal. We take the diff of consecutive proof states to determine what propositions each tactic introduces and consumes. We then construct a directed acyclic graph (DAG) whose nodes are the propositions introduced throughout the proof and whose edges encode which earlier propositions each new proposition consumes. We use this DAG to derive a dependency graph over steps of the written proof by associating propositions with proof steps using the prose-to-Lean links generated earlier. This latter graph encodes, for any given step, (1) which earlier steps it relies on, via which propositions, (2) which later steps use propositions from it, (3) the propositions this step requires, and (4) the propositions this step introduces.

\paragraph{Execution.}
To gather values from the Lean proof (Figure~\ref{fig:pipeline}.2), we first create a temporary copy of each \code{have} statement in the Lean proof and wrap it in a \code{try} block, which lets us isolate evaluation of individual steps. We prefix each \code{try} block with a hypothesis that fixes the input variable to the value chosen by the user. Because the hypotheses in the proof state encode the operations and assertions the step makes about those variables, applying the \code{subst}, \code{simp}, and \code{rfl} tactics reduces them to computed results: a hypothesis such as \code{n = x\^{}2 - 1} becomes \code{n = 3}. The resulting proof states are saved to JSON. We then prompt a language model (GPT-5.1) with prose steps and related JSON states to produce string templates (Figure~\ref{fig:pipeline}.3) for showing worked-out math. Generated templates include placeholders for values extracted from the Lean state (e.g., \code{x\^{}2 = \{\{x\}\}\^{}2 = \{\{n\_def\}\}}). Whenever a user selects new input values, our user interface re-renders the templates to reflect the Lean proof states extracted for that set of values (Figure~\ref{fig:pipeline}.4).




\begin{figure}
    \centering
    \includegraphics[width=\linewidth]{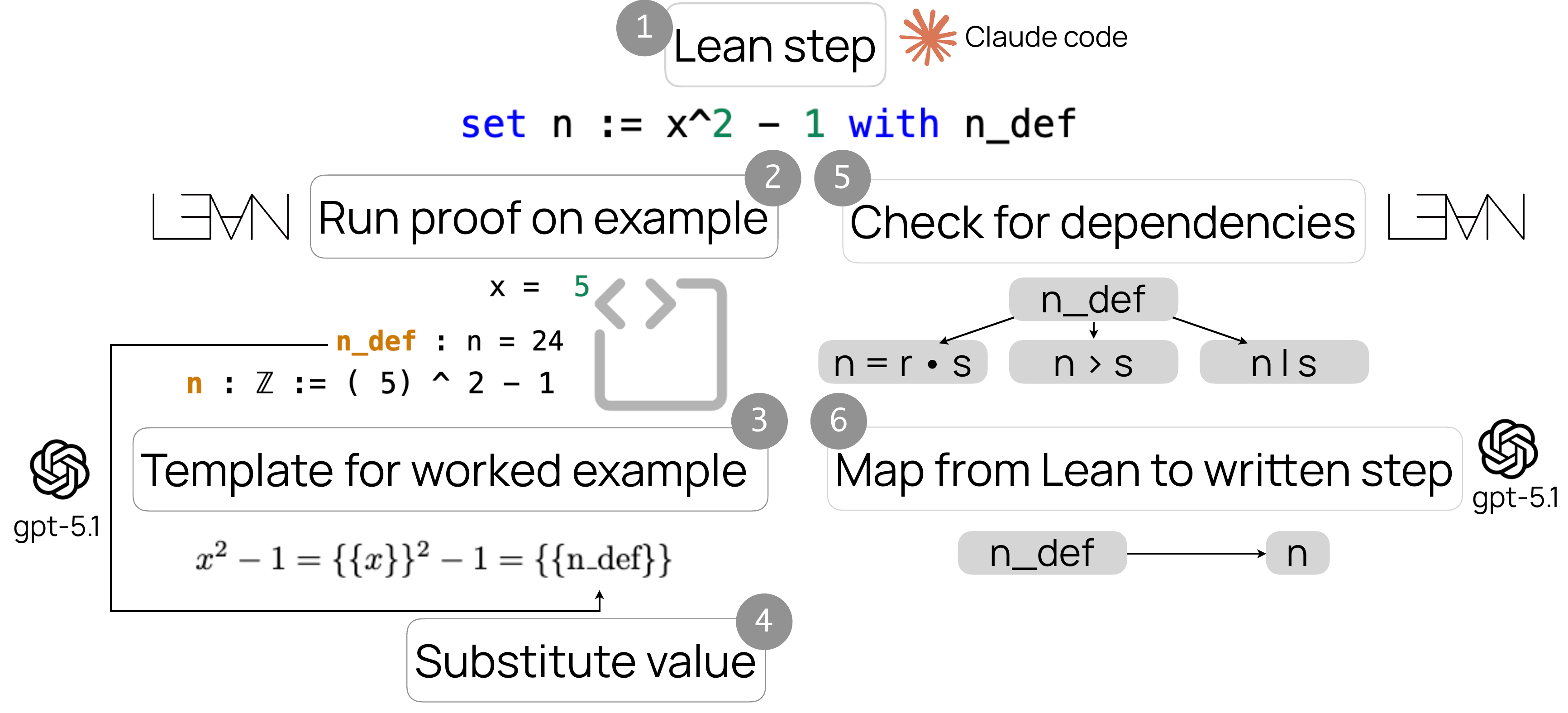}
    \caption{Technical approach to make a proof explorable. \mdseries First, a Lean proof is generated (\circlenum{1}). Then, the proof is executed on input values and Lean proof state is extracted~(\circlenum{2}). A language model then generates a template for the proof step worked out on an example~(\circlenum{3}). Lean-computed values are programmatically substituted into this template~(\circlenum{4}) to produce the verified worked example shown to the reader. The consecutive states of the proof are diffed to create a dependency graph~(\circlenum{5}). A mapping is created that links Lean statements to natural-language propositions~(\circlenum{6}).
    \Description{A diagram of the six-stage pipeline that makes a proof explorable, with numbered markers. A Lean proof step is generated by Claude Code. That step is then run on an example input to extract the Lean proof state, and separately checked for dependencies between propositions. GPT-5.1 generates a natural-language template for the worked example and a mapping from Lean variables to written proof steps. Lean-computed values are substituted into the template to produce the final worked example shown to the reader. Arrows connect the stages, and model icons mark which stages use Claude Code and which use GPT-5.1.}
    }
    \label{fig:pipeline}
\end{figure}

\subsection{Technical notes}
\label{sec:tech-eval}

Our approach works well on the four proofs used in the user study. \change{Here, we characterize its performance on one of those proofs (in Appendix~\ref{tech-notes-thm}). For this proof, 19 of the 22 dependency edges ($86\%$) were correctly recovered. Worked example templates were correctly instantiated for all steps at $n = 1$.} \change{Some human review and prompt adjustment was needed to prepare one of the proofs for the study. For the Task 1 proofs, we revised and re-ran the proof generation prompts twice before the Lean proof compiled and satisfied our structural constraints, and three times before the templates produced worked examples that were brief, accurate, and informative; the Task 2 proofs required no prompt adjustment.} These tests suggest the approach is viable for undergraduate-level proofs that are simple, algebraically tractable, and forward-reasoning. \change{More complex proofs might surface additional failure modes and require extensions to the generation and alignment, as well as interface designs for larger and less-structured parameter spaces.} 

Experiments with additional theorems have revealed some limitations in the current implementation. We describe two limitations on two number theory proofs drawn from  ProofNet~\cite{azerbayev2022proofnet}, a benchmark of undergraduate math theorem statements, natural language proofs, and formalized Lean theorem statements. We refer to the two theorems as ProofNet~1 (``for all odd $n$, $8 \mid n^2 - 1$'') and ProofNet~2 (``$3x^2 + 2 = y^2$ has no integer solutions''). 

\paragraph{Alignment.}
Mapping of natural language sentences to Lean code can fail when prose and generated Lean proofs follow different proof strategies. In ProofNet~1, for instance, the written proof factors $n^2 - 1$ as $(n+1)(n-1)$ while the Lean proof takes a different factorization approach, substituting $n = 2k+1$ in place of $n$. 

\paragraph{Dependency graph.}

Our method for computing the dependency graph, using a coding agent to extract the proof state after each Lean step and then diff consecutive states, is only an approximation. Diffing can make bookkeeping tactics like \code{rfl} appear as spurious nodes, while closing tactics like \code{omega} or \code{contradiction} discharge the goal without introducing a new proposition (as in ProofNet~2, where the last step of our generated proof had no incoming edges). A more principled approach would drop the agent entirely and analyze the underlying Lean proof terms directly. This would not have changed the behavior of the proofs in the study, but we intend to implement it, as it would matter for other cases.

\section{Study design}
To assess how these sorts of features could support understanding, we conducted a usability study on the task of proof understanding. The study compared explorable proofs to a modern alternative for explaining theorems, a written proof plus chatbot.

We structured our study around the following questions:

\begin{itemize}
    \item \textbf{RQ1}: How does the addition of the explorability features affect understanding of the original written artifact? 
    \item \textbf{RQ2}: Do readers use the explorability features and if so, how do the features shape the reading process? 
\end{itemize}

\paragraph{Participants.}
We recruited participants who were comfortable reading mathematical proofs at the undergraduate level. All 16 participants were bachelor's students in computer science, recruited from academic mailing lists at a university. 15 had completed an introductory proof course and 11 had taken advanced coursework requiring proofs. After the study, 11 participants rated themselves very familiar with proofs of the complexity used in the study, 4 rated themselves familiar, and 1 somewhat familiar.

\paragraph{Baseline}
\change{In this condition, participants were provided a static written theorem statement and proof alongside unconstrained access to the Gemini~\cite{team2023gemini} chatbot. This allowed us to differentiate the impact of our specific features from AI support generally. To reduce the time penalty of manually typing in prompts, we provided some pre-written prompts that participants could optionally copy and paste into the Gemini interface (Appendix~\ref{baseline-prompts}); these prompts mirrored the kind of support explorable proofs offer, like working out the proof on an example.} In our pilot tests, Gemini yielded correct answers, reducing the likelihood that comprehension differences are due to correctness of the mathematical content presented.

\paragraph{Procedure.}
Each participant completed a one-hour study session over Zoom. Participants accessed the web-based study interfaces through their own computers, with remote access provided through a locally hosted deployment exposed via ngrok. We collected screen recordings and either audio recordings or hand-written notes, depending on participant preference. To minimize demand characteristics, we referred to tools with neutral names (``Tool S'' for explorable proofs and ``Tool R'' for the baseline). The study used a within-subjects design: each participant completed one task with each tool. Task order was fixed across participants, while tool order was counterbalanced. Participants completed a consent form and  filled out a background questionnaire. They then completed a short tutorial and a proof-understanding task in each of the two conditions. Finally, in a closed-book setting, they answered two sets of proof-comprehension questions, one per task, and solved two proof-ordering puzzles.

\paragraph{Tutorial.}
Before beginning the tasks, participants completed a tutorial. We introduced the features of each tool and had participants complete short exercises practicing each feature. For example, for the baseline condition, this included copying and pasting prompts and prompting Gemini freely. For the explorable proofs tool, this included using the theorem input slider, viewing worked examples for each step, entering their own inputs, and tracing proof step dependencies. Each tool was introduced immediately before the task in which it was used.

\paragraph{Tasks.}
Participants completed two proof-understanding tasks, one per tool. Each task involved a set of two theorems: a first theorem with its full proof, and a similar second theorem included to test whether participants could synthesize proof ideas across related theorems. The theorems can be found in Appendix~\ref{sec:theorems-and-proofs}. In each task, participants were instructed to ``understand the theorem and its proof technique well enough to answer detailed questions about it (from memory). Treat this like preparing for a short quiz.'' Each task lasted 15 minutes. Participants began by reading the first theorem and its proof, then advanced to the second theorem and proof when they felt ready. If they had not moved on by the 10-minute mark, they were prompted to do so.\footnote{This soft time limit ensured that participants had at least 5 minutes to study the second theorem in the set in order to answer a comprehension question on generalizing  the pair of theorems and proofs. Our prototype also offered a simple feature for comparing proofs. While it appeared useful, it was not the main focus of the analysis.}

Following the proof-understanding tasks, participants answered five proof-comprehension questions per theorem set that addressed a variety of levels of \citet{mejiaramos2012assessment}'s proof comprehension inventory. These questions can be found in Appendix~\ref{questions}. The first required participants to describe the proof in their own words. Following the example of ~\citet{davies2020comparative}, we used these descriptions as proxies for comprehension in a comparative judgment assessment. The remaining four questions covered specific aspects of proof comprehension: understanding dependencies between steps, using examples to check the necessity of premises, interrogating the theorem's assumptions, and identifying shared strategies across multiple related proofs.

As an additional test of proof understanding, we had participants complete a proof-ordering puzzle~\cite{poulsen2022proof} for the first proof shown in each task. Each puzzle presented the original proof steps in scrambled order, with ordering words such as ``first'' and ``lastly'' removed and 3--4 distractor steps added. Participants were asked to arrange the steps in the correct order.

\paragraph{Questionnaires}

Participants completed three additional questionnaires. After each proof-understanding task, they completed a subset of the NASA-TLX workload assessment~\cite{hart1988development} and reported the usefulness of the tool and its individual features on a 4-point usefulness scale. \change{The 4-point scale ranged from ``not at all useful'' to ``very useful'' and was chosen for its interpretability for both participants and researchers.} At the end of the study, participants filled out a questionnaire about their experience using both tools.

\paragraph{Measurements and analysis.}

To optimize our evaluation resources, we split the grading process between two evaluators based on the nature of the assessment. A senior computer science faculty member with experience teaching proof-based mathematics ranked participants' proof description responses using a TrueSkill ranking model~\cite{herbrich2006trueskill}. This model maps a series of pairwise preference judgments to estimates of response quality.  Rather than requiring full pairwise comparisons which would demand 128 judgments (16 responses $\times$ 8 comparisons in the opposite condition), we used an adaptive design with a fixed budget of 40 comparisons per theorem, yielding approximately 5 comparisons per response. Pair selection was adaptive: responses were initialized with equal ratings, and at each step the algorithm selected the most closely matched pair for comparison, updating quality estimates via the TrueSkill ranking model after each judgment. The expert was blind to which condition each description was written in. Because the fixed comparison budget can leave item-level posteriors partially converged, we assessed the robustness of the condition-level contrast in two ways: sampling from the fitted posteriors, and refitting TrueSkill from the prior on 1{,}000 random shuffles of the judgment sequence.

For the fine-grained response evaluation, a computer science graduate student with prior discrete math teaching experience evaluated responses to the remaining 4 of 5 proof questions along three rubric dimensions adapted from \citet{mejiaramos2012assessment}'s proof assessment model: (1) correctness, (2) reference to specific steps of the proof, and (3) reference to concrete examples. The grader was blind to condition. The same grader also performed a comparative ranking of responses to the proof description question, pooled across tasks and conditions, along dimensions of mathematical understanding, level of detail, and clarity, assigning each response to one of five evenly distributed tiers. This forced distribution allows tier assignments to serve as a relative measure of quality. \change{To assess reliability for the detail, understanding, and clarity codes, we recruited a second grader with a similar profile to independently grade the Task 1 responses. In the results, we report the first grader's ratings alongside interrater reliability measured using quadratic-weighted Cohen's $\kappa$.}

All statistical tests were performed as two-tailed comparisons with significance threshold $\alpha = .05$.
We used Mann--Whitney $U$ tests for between-subjects comparisons of correctness scores (0--3 ordinal scale), rubric dimension ratings, and estimated TrueSkill response quality $\mu$, reporting Cohen's $d$ as the effect size. For binary outcomes (whether participants referenced specific proof steps or included concrete examples in their proof descriptions), we used Fisher's exact tests.
For the proof-ordering puzzle, we compared accuracy, task duration, and number of actions between conditions using paired $t$-tests with $d_z$ as effect size. For Likert-scale responses we used two-tailed Wilcoxon signed-rank tests, and tool preference was evaluated with a binomial test against chance ($p = 0.5$).

When making multiple comparisons, we applied Holm--Bonferroni corrections within each family of tests, maintaining a family-wise error rate of $\alpha = .05$. When reporting effect sizes, we report 95\% confidence intervals from 20{,}000 bootstrap resamples.



We analyzed open-ended responses through thematic analysis~\cite{blandford2016qualitative}: one author performed open coding to identify recurring ideas, which were then refined into key themes. We also conducted a video analysis of interaction patterns for the 12 participants who consented to recording the study session.\footnote{Four participants did not consent to video recording and are excluded from interaction-level analyses.} \change{In Figure~\ref{fig:timeline}, for each participant one author logged three interaction types by cursor position and clicks: the cursor on a step in the original written proof, the cursor on the worked-out example, and a click on a proof step to reveal its dependencies.}


\section{Explorable proofs...}

Here, we report the results from our usability study.

\begin{figure}
    \centering
    \includegraphics[width=\linewidth]{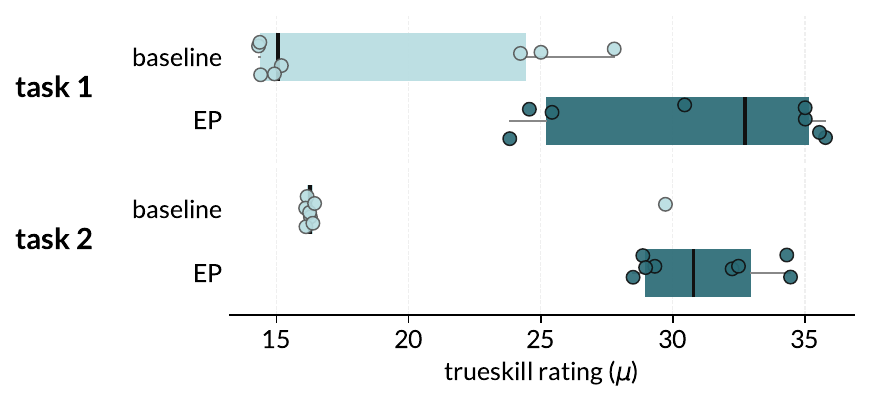}
    \vspace{-5ex}
    \caption{Estimated TrueSkill proof description quality by task and condition. \mdseries Box plot shows the distribution of TrueSkill $\mu$ estimates assigned to participant responses through pairwise comparisons by a mathematics instructor, per task.}
    \label{fig:pairwise}
    \Description{A box plot of estimated TrueSkill ratings for participants' proof descriptions, split into two panels for Task 1 and Task 2. Each panel compares the baseline condition with the explorable proofs condition, with individual participant scores overlaid as points. In both tasks the explorable proofs distribution sits substantially higher and its box spans a narrower range than the baseline. The horizontal axis is TrueSkill rating, ranging from about 15 to 35.}
\end{figure}

\subsection{...improve proof comprehension (RQ1)}

Proof descriptions were consistently ranked better by an instructor when written after using explorable proofs (hereafter ``EP'') than after using the baseline. A mathematics instructor preferred proof descriptions written with explorable proofs significantly more often than those written with the baseline (Figure~\ref{fig:pairwise}), and these descriptions also received significantly higher TrueSkill ratings across both study tasks (Task 1: EP $M=30.70$ vs.\ Baseline $M=18.78$, preferred in $77.5\%$ of pairwise comparisons, $U=58.0$, $p=.005$, $d=2.14$; Task 2: EP $M=31.14$ vs.\ Baseline $M=17.93$, preferred in $90.0\%$ of pairwise comparisons, $U=60.0$, $p=.002$, $d=3.47$). \change{Because the fixed comparison budget left item-level posteriors partially converged (median $\sigma = 5.17$ from an initial $8.33$), we checked the robustness of the condition-level contrast: it held under posterior sampling ($\mu_{EP}$ - $\mu_{Baseline}$ = $+12.6$, 95\% CrI $[+9.1,+16.1]$) and under refitting TrueSkill on 1{,}000 random shuffles of the 80 judgments ($+11.8$).}

Consider the following example. Participants were asked to describe the proof of this theorem, ``For all integers $n$, $n(n+1)(n+2)(n+3)+1$ is a perfect square'', in their own words.

\paragraph{Answer A (with explorable proofs).}
\begin{quote}
$p =$ the product of the first and last term $n(n+3)$, $(n+1)(n+2)$ becomes $p+2$. multiply the outer numbers $n \cdot (n+3) = n^2 + 3n$ to be $p$ and inner numbers $(n+1)\cdot(n+2) = n^2 + 3n + 2$ to be $q$. note that $q = p + 2$. so then $p \cdot q = p \cdot (p + 2) = p^2 + 2p$, add 1 = $p^2 + 2p + 1 = (p + 1)^2$. then $p = n^2 + 3n$, substitute this, and you get the whole thing $(p + 1)^2 = (n^2+3n+1)^2$ this is a perfect square.
\end{quote}

\paragraph{Answer B (with baseline).}
\begin{quote}
We group the middle terms together and the outer terms together, forming $p$ and $q$. Then we write $q$ in terms of $p$ and factor the remaining expression, which shows that it is a perfect square in terms of $p$. We then substitute the expression for $p$ back in to get a perfect square in terms of $n$.
\end{quote}

When asked which answer was better in terms of answer quality, the instructor chose Answer A and explained:

\begin{quote}
I would give Answer A full points. Answer B is imprecise and can be interpreted the wrong way\ldots{}students must learn to formulate math answers precisely, and using math notation is necessary for that.
\end{quote}

When the graduate student did an expanded assessment, they reported proof descriptions as including significantly more explicit proof steps when written in the explorable proofs condition for Task 1 ($100\%$ vs $38\%$, $p=.026$) though not for Task 2 ($100\%$ vs $50\%$, $p=.077$). They also marked descriptions from the explorable proofs condition as significantly more detailed (Task 1: $U=53.5$, $p=.025$, $d=1.41$; Task 2: $U=53.5$, $p=.025$). \change{Detail ratings appeared reliable, with strong agreement between graders (quadratic-weighted $\kappa = .82$).} The graduate student's ratings for understanding and clarity were similar to those for level of detail.\footnote{Overall, the graduate student's ratings were highly correlated across these dimensions, which may be an indication that they reflect a single latent quality factor.} \change{Agreement on understanding and clarity was lower, suggesting lower reliability (quadratic-weighted $\kappa = .56$ and $.44$, respectively).} Proof descriptions were graded to be significantly more correct for Task 2 ($U=54.0$, $p=.012$) though not for Task 1 ($U=43.0$, $p=.216$).


\begin{figure}
    \centering
    \includegraphics[width=\linewidth]{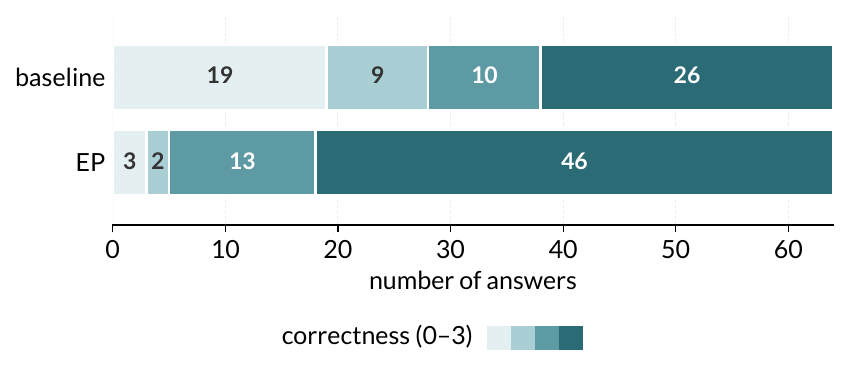}
    \vspace{-4ex}
    \caption{Response correctness by condition. \mdseries Stacked bars show counts of correctness scores (0--3) for 64 answers to proof comprehension questions from each condition, covering all questions except the proof description task. Responses written after using explorable proofs were concentrated at the highest correctness scores, while baseline responses were spread across the full range.}
    \Description{A horizontal stacked bar chart comparing correctness scores for 64 proof comprehension answers in each condition. Scores range from 0 to 3, shown as shaded segments from light to dark. The baseline bar is spread across the full range, with 19 answers scoring 0, 9 scoring 1, 10 scoring 2, and 26 scoring 3. The explorable proofs bar is concentrated at the top score, with 3 answers scoring 0, 2 scoring 1, 13 scoring 2, and 46 scoring 3.}
    \vspace{-1em}
    \label{fig:proof-grades}
\end{figure}

For the remaining proof questions besides the proof description, responses in the explorable proofs condition were graded as more correct than the baseline (EP $M=2.59$ vs.\ Baseline $M=1.67$; Wilcoxon signed-rank $W=4.0$, $p=.001$, Figure~\ref{fig:proof-grades}). They also referenced explicit proof steps more often (proportion of each participant's proof comprehension-question responses that referenced explicit proof steps, averaged across participants, EP $38\%$ vs.\ Baseline $11\%$, $W=14.0$, $p=.014$).

Consider the following examples of responses from participants who received the median correctness score in their respective conditions. In response to a question about how the theorem ``For all integers $n$, $n(n+1)(n+2)(n+3)+1$ is a perfect square'' could be generalized, one participant, after using explorable proofs, wrote:\footnote{The proof pairs the outer terms ($n$, $n+3$) against the inner terms ($n+1$, $n+2$). Both pairs have the same sum, so their products differ by an amount that does not depend on $n$: $n(n+3)=n^2+3n$ and $(n+1)(n+2)=n^2+3n+2$ differ by exactly $2$, whatever $n$ is.}

\begin{quote}
``Any case where you're able to pair the outer terms and the inner terms, the second product changes from the first by a constant. Then, if you add the right constant to pad that expression, the whole expression becomes a square. So the theorem can be generalized by changing the spacing and adjusting the added constant.'' 
\end{quote}

A participant in the baseline condition wrote:

\begin{quote}
``Yes, as long as the pairs of terms have a difference of a constant there's a $+1$ in the end.''
\end{quote}

The response in the explorable proofs condition correctly identifies the structural mechanism of the proof. If the two products differ by a constant,\footnote{The participant pairs the outer terms ($n$, $n+3$) and the inner terms ($n+1$, $n+2$). Expanding this gives $n(n+3) = n^2+3n$ and $(n+1)(n+2) = n^2+3n+2$. The $n$-dependent parts are the same. The offsets within each pair add to the same total ($0 + 3 = 1 + 2$), so subtracting one from the other leaves only the constant $2$.} they can be written as $m-c$ and $m+c$, so their product is $m^2-c^2$ and adding $c^2$ makes it a square. The added constant is the square of half the difference between the products: $1$ for consecutive integers, and $r^4$ for spacing $r$, where the difference is $2r^2$. The response in the baseline gestures at the right idea (constant difference between pairs) but treats the $+1$ as fixed, missing that the added constant is determined by the spacing.

\subsection{...do not affect the ability to recall proof order (RQ1)}

Participants performed similarly in both conditions on the proof-ordering puzzle. Pairwise accuracy, task duration, and number of actions did not differ significantly between conditions.

Proof accuracy, the proportion of proof steps in their correct position, was similar across conditions (EP: $M=59.0\%$, $SD=33.1\%$; Baseline: $M=63.9\%$, $SD=32.8\%$) and did not differ significantly (paired $t(15)=-0.408$, $p=.689$, $d_z=-0.10$, 95\% CI $[-0.72, 0.38]$). Pairwise accuracy, the proportion of proof-step pairs placed in the correct relative order to one another, was high in both conditions (EP: $M=93.2\%$, $SD=11.2\%$; Baseline: $M=95.1\%$, $SD=7.3\%$) and did not differ significantly (paired $t(15)=-0.609$, $p=.551$, $d_z=-0.15$, 95\% CI $[-0.62, 0.39]$). Task duration was also similar across conditions (EP: $M=148.7$s, $SD=65.4$; Baseline: $M=157.7$s, $SD=50.0$; $t(15)=-0.559$, $p=.585$, $d_z=-0.14$, 95\% CI $[-0.73, 0.37]$). Likewise, the number of actions did not differ significantly (EP: $M=33.4$, $SD=7.5$; Baseline: $M=31.3$, $SD=8.9$; $t(15)=0.693$, $p=.499$, $d_z=0.17$, 95\% CI $[-0.42, 0.62]$).

\subsection{...enable useful interactions (RQ2)}

Participants reported higher perceived success with explorable proofs (Figure~\ref{fig:likert}; EP: $M=6.1$, $SD=0.7$; Baseline: $M=5.3$, $SD=1.2$ on a 1–7 scale; $W=3.5$, $p=0.012$) and felt less rushed (EP: $M=2.3$, $SD=1.0$; Baseline: $M=3.1$, $SD=1.3$; $W=3.5$, $p=0.007$). 13 of 16 participants (81\%) preferred explorable proofs, and 3 of 16 (19\%) preferred the baseline (binomial test, $p=0.021$). There was no significant difference in which tool participants said they would use in the future for understanding math theorems (EP: $M=5.6$, $SD=1.5$; Baseline: $M=5.0$, $SD=1.0$; $W=23.5$, $p=0.114$) or in perceived mental demand (EP: $M=3.1$, $SD=1.7$; Baseline: $M=3.6$, $SD=1.4$; $W=6.5$, $p=.099$). Applying a Holm--Bonferroni correction across the four paired Likert comparisons, felt rushed and felt successful remained significant (Holm-adjusted $p=.028$ and $p=.036$), while mental demand and desire to use the tool in the future remained insignificant.

\begin{figure}
    \centering
    \includegraphics[width=\linewidth]{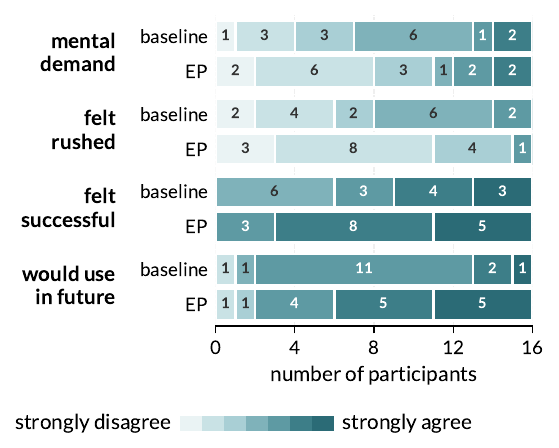}
    \vspace{-4ex}
    \caption{\textbf{Comparisons in self-reported experience across condition.} \mdseries Stacked bars show how many participants gave each rating on a 1--7 agreement scale, from strongly disagree (light) to strongly agree (dark). Participants felt significantly less rushed and more successful with explorable proofs. No significant differences were found for mental demand or desire to use the tool in the future.}
    \Description{Four pairs of horizontal stacked bars comparing self-reported experience between the baseline and explorable proofs conditions, one pair each for mental demand, felt rushed, felt successful, and would use in future. Each bar shows how many of the 16 participants gave each rating on a 1 to 7 agreement scale, shaded from light for strongly disagree to dark for strongly agree. Ratings for felt rushed shift toward disagreement and ratings for felt successful shift toward agreement in the explorable proofs condition, while mental demand and would use in future look similar across conditions.}
    \label{fig:likert}
    \vspace{-1em}
\end{figure}

\paragraph{Feature usefulness.}
On the 4-point usefulness scale (1=``Not at all useful'' to 4=``Very useful''), the most useful feature was seeing mathematical expressions worked out on examples ($M=3.6$, $SD=0.9$), followed by reviewing valid and invalid ranges of values using sliders ($M=3.6$, $SD=0.7$), and checking whether the theorem holds on a specific example ($M=3.4$, $SD=0.6$). Participants also found it useful to look up prior dependent steps ($M=3.3$, $SD=0.9$) and check whether an individual proof step holds on an example ($M=3.0$, $SD=0.7$). Overall, these ratings suggest that participants found example features the most useful, though features for exploring proof structure were also useful.

\begin{figure}
    \centering
    \includegraphics[width=\linewidth]{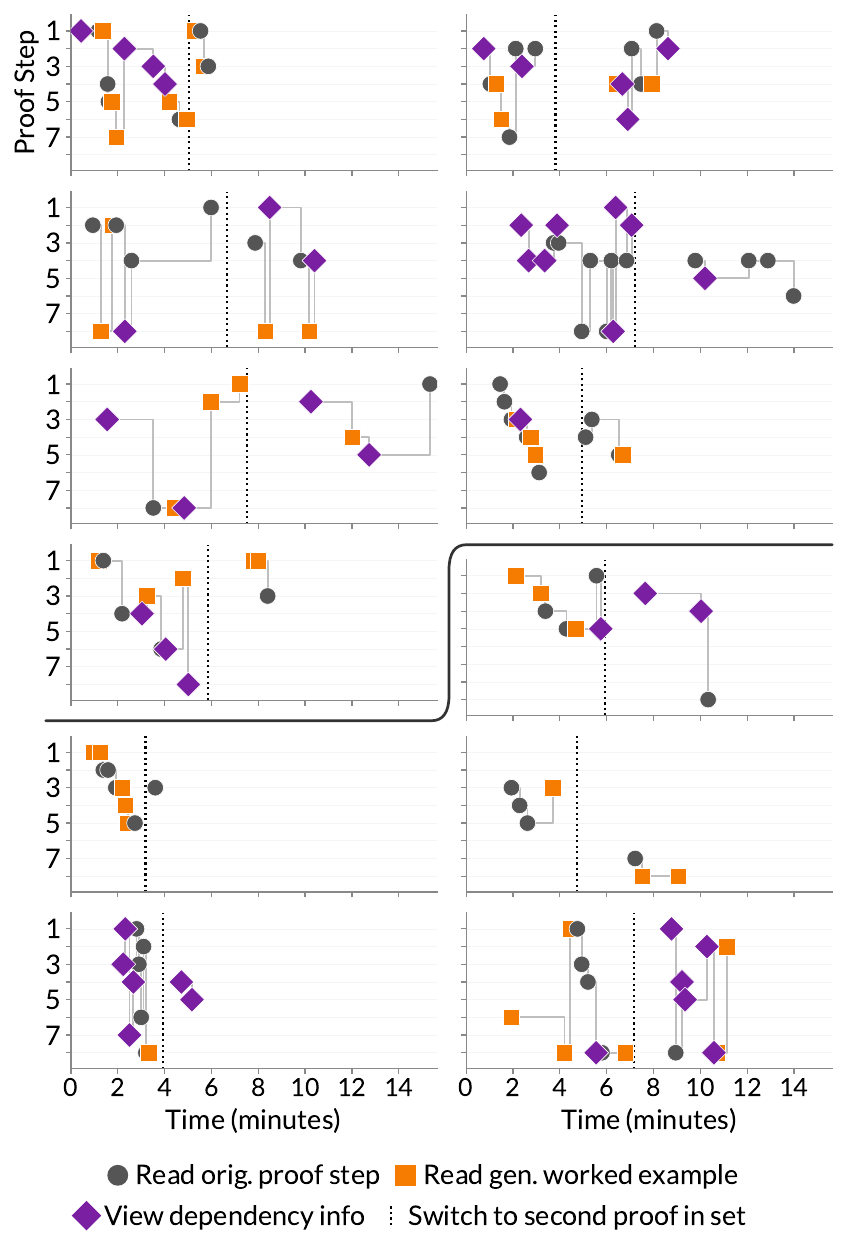}
    \vspace{-5ex}
    \caption{Timelines of interactions with explorable proofs. \mdseries Each subplot shows the interactions for one participant in the explorable proofs condition. The 7 above the black separator were participants completing Task 1, and the 5 below were completing Task 2. The y-axis indicates which proof step (1--8) the participant is engaging with. Gray circles and orange squares denote reading in the original proof steps and the corresponding generated worked examples, respectively. Purple diamonds indicate clicks to open the dependency widget. The x-axis is time within the task. Within each task were two proofs. A dotted line marks when each participant transitioned between proofs.
    }
    \Description{A grid of twelve small timeline plots, one per participant in the explorable proofs condition, arranged in six rows of two. A stepped black line separates the seven Task 1 participants from the five Task 2 participants. In each plot the vertical axis is the proof step, numbered 1 through 8, and the horizontal axis is elapsed time from 0 to about 15 minutes. Gray circles mark reading an original proof step, orange squares mark reading the corresponding generated worked example, and purple diamonds mark clicks that open the dependency widget. A vertical dotted line marks when the participant moved from the first proof to the second. Markers frequently move upward as well as downward, showing participants revisiting earlier steps, and gray and orange markers often interleave.}
    \vspace{-1em}
    \label{fig:timeline}
\end{figure}

\paragraph{Patterns of interaction.}

We detail how participants used Gemini, which reveals what kinds of proof understanding they were seeking. In the baseline, participants asked Gemini an average of $2.8$ questions ($SD = 2.4$, range $0$--$9$). The 44 total queries fell into three broad categories: requesting proof descriptions or explanations (17 queries), asking about specific numeric examples worked out on the proof (18), and asking about proof structure such as individual steps, hypotheses, or conclusions (9). One participant's questions were particularly thorough (P10), with 6 about proof steps, hypotheses, and conclusions, such as one question about why a specific ordering constraint was necessary. One participant (P14) did not use Gemini at all and opted to only read the written proof. That participants asked Gemini about examples and proof structure underscores the importance of related features in explorable proofs.

Figure~\ref{fig:timeline} shows a preliminary account of participants' proof reading behaviors when using explorable proofs. Participants made heavy use of the examples, selecting on average $10.8$ examples ($SD = 8.1$, counting examples that were selected for at least five seconds). Most were reached by using the theorem input slider: $5.6$ ($SD = 3.6$) satisfied the theorem's hypotheses and $3.8$ ($SD = 5.1$) did not. The remaining $1.4$ ($SD = 1.1$) were inputs participants entered themselves. The utility of examples is articulated by P2 who told us that ``this helps me understand the value [of variable x]\ldots{}and I know what [the proof is] trying to show.''

At the start of the session, participants frequently did something resembling stress testing: they selected example inputs that violated the theorem's assumptions to start reading the proof on. Nine participants did so using the slider, and three by inputting their own examples. P10 described wanting to ``try some things where it's an edge case [and] something that's an intuitive value'' and then ``check out\ldots{} values that are less intuitive'' and see whether the theorem still held. 


Participants made use of the flexibility afforded by explorable proofs. Eleven participants read proof steps out of order at least once, as demonstrated by lines in Figure~\ref{fig:timeline} going up and down. P3, for example, read forward through steps 4, 6, and 7 before jumping back to step 2. Ten participants used the dependency links at least once. P5 described dependency links as ``the ability to track the history and future references of a proof or immediate step.'' Participants valued this ability: P9 noted they could ``easily track how the proof got certain variables by clicking on the [dependency link],'' and P10 appreciated seeing ``how each part of the proof interacted with the rest of it.''

Figure~\ref{fig:timeline} shows the different high-level approaches that participants took when reading proofs. Some participants referred heavily to examples, and others to dependency links. This suggests that explorable proofs may support different reader tendencies, both example-driven and structure-driven, rather than funneling readers into one particular way of reading.

Seven participants alternated between reading the original proof text and generated worked examples within the same session, as visible in the interleaving of gray and orange markers in Figure~\ref{fig:timeline}. P2 switched modes 8 times across 17 step visits, while P3 switched 6 times across 16 visits. This suggests participants used the worked examples not in a separate reading pass but integrated with their review of the proof. P2 described how they integrated these two sources of information, stating ``the text is showing the meat of the proof and the boxes are showing how it works out on an example\ldots{}I'm doing two things in parallel.''


\section{Discussion}

Our study found that participants who used explorable proofs produced proof descriptions that were more frequently preferred by an instructor. They also answered proof comprehension questions that were more correct and detailed. Participants used its features, found them valuable, and engaged in a range of exploratory behaviors.

\subsection{Why did explorable proofs have an effect?}

What mechanisms might explain the better outcomes we observed in the explorable proofs condition? Several are plausible and possibly act in concert. The novelty of the interface may have increased engagement. The concision of worked examples in explorable proofs, showing only the values relevant to each step, may have made them easier to read than more long-winded chatbot answers. And dependency links could have helped participants build a more accurate mental model of proof structure. Participants' very frequent use of examples and strong ratings of their utility suggest that example-based exploration features were natural and desirable. \change{One significant factor may have been the speed with which participants could set up and switch between examples. We observed a large difference in interaction volume between conditions, with participants triggering an average of $10.8$ examples in explorable proofs while submitting only $2.8$ queries with the baseline. It could be that velocity of interaction, rather than the formal backing per se, drove better outcomes. In our view, however, interaction volume and formal backing are not competing explanations but one mechanism: exploitation of the formal backing is what made rapid, low-cost interaction possible.}

\subsection{A blueprint for formally-backed interactive explanations}
\label{sec:incidental-diff}

While we explore formally-backed explanation interfaces in the context of mathematical proofs, the approach may generalize. We reflect on the design decisions that we believe drove the effect, and distill them into principles for others building similar systems.

\paragraph{Offload incidental difficulty.}
One might ask whether interactivity removes the productive struggle necessary for durable learning~\cite{bjork2011making}. Proof comprehension involves two distinct kinds of difficulty: \textit{essential} difficulty (such as understanding why a step follows or is necessary) and \textit{incidental} difficulty (such as tracking variables, mentally substituting values, recalling which step a symbol came from). We designed explorable proofs to remove incidental difficulty while preserving the essential difficulty. We believe this is why participants felt less rushed without a reduction in mental demand.

\change{However, what counts as incidental is reader-dependent: for a novice still building symbolic fluency, carrying out substitutions by hand may itself be formative rather than mere overhead, in which case automating it could remove productive struggle rather than free attention for it. Our current study cannot adjudicate this: it does not measure delayed retention that would reveal whether reducing this friction aided learning long-term.}

\paragraph{Stay close to the source.}
With explorable proofs, interactions are enabled by connecting to a formal proof object. Worked example values are extracted from the Lean proof state; dependency edges are recovered by diffing Lean proof states; the worked example template mirrors the structure of the written step. This closeness to the source matters because distance introduces error: whenever AI is used to create a model or establish links, there is opportunity for mistranslation. We reduced this distance through several design decisions: prompting the agent to generate a Lean proof that structurally mirrors the prose and performing linking in small increments that leverage structural bridges put in place during generation. Calculation was delegated to the formal model. We believe similar strategies, minimizing the role of AI inference and maximizing the role of the formal artifact, could reduce distance in other formally-backed explanation systems, such as tools for understanding programs or circuit diagrams.

\subsection{Limitations}

Our technical approach has some limitations. First, because the Lean proof is AI-generated, it is not guaranteed to follow the original prose's logical path; agents sometimes produce structurally different arguments, or fail to complete the proof. Second, stages of our approach rely on LLMs which introduce the possibility of subtle errors through mistranslations from Lean to natural language. However, as AI tools for formal proof generation improve, agents will become better at producing Lean proofs that structurally mirror their prose counterparts.

\change{Different kinds of failures carry different costs for readers. When dependency connections are incorrectly recovered, readers may waste time tracing propositions to irrelevant steps. When templates or Lean-state links are wrong, readers may be shown incorrect worked-out math. This is likely the worst case, since the interface's formal backing may lend such errors unwarranted credibility. Were explorable proofs deployed, we see two mitigations: giving readers the standard guidance to check the math in LLM generations, and, for course-maintained materials, enabling teaching staff to spot-check critical proofs before release.}

Our study also has limitations. We focused on relatively short proofs with undergraduate computer science students; the benefits of explorable proofs may differ for longer proofs, more advanced students, or experts encountering unfamiliar material. Additionally, while our evaluation measured some proxies of comprehension, it does not measure long-term retention, transfer to new problems, or the ability to independently construct related proofs.

\subsection{Future work}
\label{sec:future-work}

Future work could further develop explorable proofs in three directions. First, systems could combine explorable interfaces with conversational interaction. Rather than replacing the formal backing, chat would operate over it, allowing users to ask what breaks if a hypothesis is removed or how a proof changes under different instantiations. Second, the approach could be extended to broader mathematical content to understand how it holds up in messier, realistic settings. This includes handling research-level theorems lacking complete Lean proofs, complex structures that do not map cleanly to natural language, and non-numeric inputs like graphs. Finally, the principle of formally-backed explorability could extend beyond mathematics. As formal verification tools are applied to end-user plans~\cite{lee2025veriplan} and complex reasoning~\cite{jiang2024leanreasoner}, future systems might pair everyday artifacts, like a tax document or a plan that has many constraints, with a formalized Lean counterpart to make them fully explorable. 

\change{Future evaluations could also compare explorable proofs to a no-technology baseline, investigating whether the interface offloads incidental difficulty as intended or removes productive struggle that unassisted readers would otherwise benefit from. More broadly, the boundary between essential and incidental difficulty should be further studied: work relating these categories to germane and extraneous load in cognitive load theory~\cite{sweller2011cognitive} could clarify which frictions in proof reading should be elided.}
\section{Conclusion}


In this paper, we explore how to make explanations explorable by connecting them to formal backings. We instantiate the idea in \textit{explorable proofs}. In explorable proofs, a natural language proof is linked to its formal Lean counterpart. Analysis of proof state yields information that plugs into the written proof to allow readers to inspect steps worked out on concrete examples and trace logical dependencies. In a lab study, participants who used explorable proofs wrote proof descriptions that an instructor preferred over those written with a chatbot baseline. 

\begin{acks}
We are grateful to Val Tannen, Xiaorui Liu, and Daniel Feschbach for grading participant responses; extra gratitude to Val for providing motivating examples and to Xiaorui for his feedback. We also thank Jeff Tao for his thoughtful feedback on drafts of this paper and Ashwin Padaki for his timely review of the TrueSkill algorithm rankings.
This research was developed with funding from the Defense Advanced Research Projects Agency's (DARPA) SciFy program (Agreement No. HR00112520300) and the National Science Foundation under Grant Number \#2541577. The views expressed are those of the authors and do not reflect the official policy or position of the Department of Defense or the U.S. Government.

\end{acks}

\bibliographystyle{ACM-Reference-Format}
\bibliography{refs,andrew-base,andrew-extras}

\appendix
\section{Appendix}

\subsection{Narrative Scenario}
\label{sec:narrative-scenario}

We demonstrate how the interaction patterns in explorable proofs can support proof comprehension with an example of its usage in understanding a new proof. Ram is a sophomore in a discrete mathematics class. He is trying to understand a theorem from the elementary number theory section. Beyond understanding this particular theorem, Ram wants to build intuition for patterns in this proof that may generalize to others.

Ram first reads the theorem statement:

\begin{quote}
For all integers $x$ and odd integers $n > 1$, if $x > 1$, then $x^n + 1$ is not prime.
\end{quote}

It is not immediately clear what this theorem is saying or why it should be true. Ram loads the theorem into the explorable proofs interface to look into this.

\paragraph{Inspecting the boundaries of a theorem with the theorem input slider.}
Ram first uses the \textit{theorem input slider} to see whether the theorem's conclusion holds on specific integers. He sets $x = 2$ and $n = 3$, and the slider shows that the assumptions are satisfied. The tool displays the worked example below: $x^n + 1 = 2^3 + 1 = 9$, which is indeed not prime. Seeing a valid example helps the theorem click, but Ram wonders why the strict conditions, like $n > 1$ and $x > 1$, are necessary. He uses the theorem input slider to test these boundaries. When Ram slides the indicator down to $x = 1$ and $n = 1$, the interface immediately flags that the assumptions are no longer met. The computation updates to show $1^1 + 1 = 2$, which is prime, breaking the theorem's conclusion. Ram can now clearly see the theorem's boundaries: the logic fails exactly when $n \leq 1$ or $x \leq 1$. He understands what the theorem claims and when it applies, but he still does not know why the valid cases always work. He begins to read the proof.

\paragraph{Seeing a worked example on each proof step.}


The proof's strategy is to show that $N = x^n + 1$ is never prime by factoring it into two smaller integers, defined in the text as $r$ and $s$. To understand how these factors are constructed, Ram clicks on their definitions. The \textit{proof step dependencies} show the definition of these terms: $r = x + 1$ and $s$ is an alternating sum: $s = x^{n-1} - x^{n-2} + \cdots - x + 1$. It is not immediately clear why multiplying these two specific expressions will always equal $x^n + 1$. For his running example ($x = 2, n = 3$), the \textit{worked example} beneath the step shows $r = 2 + 1 = 3$ and $s = 2^2 - 2 + 1 = 3$. Since $3 \times 3 = 9$, the factorization checks out for this specific case. But Ram wants to see the structural pattern. He slides the inputs to $x = 3, n = 3$ and sees $N = 28$, with factors $r = 4$ and $s = 7$. The worked example shows the expanded form: $3^3 + 1 = (3 + 1)(3^2 - 3 + 1)$. He stares at this expansion. If he distributes the $(3 + 1)$ across the $(3^2 - 3 + 1)$, he gets $3^3 - 3^2 + 3 + 3^2 - 3 + 1$. The $-3^2$ and $+3^2$ cancel, as well as the $+3$ and $-3$. All that is left is $3^3 + 1$. The concrete numbers make the abstract trick obvious: the signs in $s$ alternate so that when multiplied by $r = x + 1$, every middle term cancels out. Ram keeps sliding the inputs and watching the worked example update, and the cancellation pattern holds every time.

\paragraph{Following a worked example to understand a key condition.}
Ram reaches the most complex step of the proof, which states that $s > 1$: since $n$ is odd, the alternating sum $s = x^{n-1} - x^{n-2} + \cdots - x + 1$ has an even number of powers of $x$ followed by a final $+1$, so the terms can be grouped as $s = (x^{n-1} - x^{n-2}) + \cdots + (x^2 - x) + 1$. Each grouped pair satisfies $x^m - x^{m-1} = x^{m-1}(x-1) > 0$ because $x > 1$, so $s$ is a sum of positive terms plus $1$, giving $s > 1$. This is far too dense for Ram to make sense of.

Ram follows the proof step on a worked example. For $n = 3$, $x = 2$, he sees $s = (2^2 - 2) + 1 = 2 + 1 = 3 > 1$. He slides to $n = 5$ and sees $s = (2^4 - 2^3) + (2^2 - 2) + 1 = 8 + 2 + 1 = 11$. Where $n = 3$ gave one grouped pair, $n = 5$ gives two, each still positive, with the same trailing $+1$. He slides through $n = 5, 7, 9$ and watches the number of pairs increase while the $+1$ remains the same.


Ram then toggles to $n = 1$ using the theorem input slider and sees the proof break at this step: with $n = 1$, $s = 1$, and the condition $s > 1$ no longer holds. The system flags that the proof breaks here. This makes clear why $n$ must be both odd and greater than $1$.

\paragraph{Seeing the proof through to the end.}
The remaining steps show that $N>s$, since $N=r\cdot s$ and $r>1$. Combined with $s>1$ from the previous step, this establishes that $s$ is a proper factor of $N$, so $N$ is not prime.

After his read of the proof, Ram understands not just the high-level proof strategy (find a proper factor to show $N$ is not prime), but the role each condition plays. The factorization requires odd $n$: only then does $x + 1$ divide $x^n + 1$, and only then does $s$ end in $+1$ with an even number of powers of $x$ before it. Each grouped pair, $x^{m} - x^{m-1} = x^{m-1}(x - 1)$, is positive when $x > 1$, and $n > 1$ gives at least one such pair, so $s > 1$.






\section{Study Materials}
\label{sec:study-materials}

\subsection{Prompt given to participants in the baseline}
\label{baseline-prompts}

\begin{verbatim}
Theorem: {{theorem}}
Proof: {{proof}}
Value: {{value}}

Explain how this specific value satisfies the theorem.
Walk through the proof step by step, substituting the
value at each step to show concretely why the
conclusion holds.
\end{verbatim}

\subsection{Theorems and Proofs}
\label{sec:theorems-and-proofs}

\subsubsection{Task 1: Composite Numbers}
\label{task-1}

\begin{theorem}
For all integers $x$, if $x > 2$, then $x^2 - 1$ is not prime.
\end{theorem}

\begin{proof}
Let $x$ be an arbitrary integer. Assume that $x > 2$. Let us denote $x^2 - 1$ by $n$. We wish to show that $n$ is not prime. We will do that by showing that $n$ has a factor that is neither 1 nor equal to $n$.

First observe that $x^2 - 1 = (x-1)(x+1)$. Let us also denote $x - 1$ by $r$ and $x + 1$ by $s$. Note that both $r$ and $s$ are factors of $n$, since $n = r \cdot s = s \cdot r$.

Now, because $x > 2$, we have $r = x - 1 > 1$. Moreover, we show that $s > 1$ by the following reasoning: $x > 2$, so $x + 1 > 3$, so $s > 1$.

Finally, multiply both sides of $r > 1$ with $s$. We get $r \cdot s > s$. However $r \cdot s = n$. Therefore we have derived $n > s$. Since $n$ has a factor $s$ such that $1 < s < n$, $n$ is composite.
\end{proof}

\begin{theorem}
For all integers $x$ and odd integers $n > 1$, if $x > 1$, then $x^n + 1$ is not prime.
\end{theorem}

\begin{proof}
Let $x$ be an arbitrary integer and $n$ be an odd integer with $n > 1$. Assume that $x > 1$. Let us denote $x^n + 1$ by $N$. We wish to show that $N$ is not prime. We will do that by showing that $N$ has a factor that is neither 1 nor equal to $N$.

First observe that for odd $n > 1$, $x^n + 1 = (x+1)(x^{n-1} - x^{n-2} + x^{n-3} - \cdots + 1)$. Let us also denote $x + 1$ by $r$ and the alternating sum $x^{n-1} - x^{n-2} + \cdots + 1$ by $s$. Note that both $r$ and $s$ are factors of $N$, since $N = r \cdot s = s \cdot r$.

Now, because $x > 1$, we have $r = x + 1 > 2$. Moreover, we show that $s > 1$ by the following reasoning: $s$ has $n$ terms alternating in sign, and since $n$ is odd, the last term is $+1$, and each positive term $x^k \geq 1$ for $x > 1$, so $s \geq 1$, and since $x > 1$ we get $s > 1$.

Finally, multiply both sides of $r > 1$ with $s$. We get $r \cdot s > s$. However $r \cdot s = N$. Therefore we have derived $N > s$. Since $N$ has a factor $s$ such that $1 < s < N$, $N$ is composite.
\end{proof}

\subsubsection{Task 2: Arithmetic Progression Squares}
\label{tech-notes-thm}

\begin{theorem}
For all integers $n$, $n(n+1)(n+2)(n+3) + 1$ is a perfect square.
\end{theorem}

\begin{proof}
Let $n$ be an arbitrary integer. Pair the outer terms: let $p = n \cdot (n+3)$. Expanding, $p = n^2 + 3n$. Pair the inner terms: let $q = (n+1) \cdot (n+2)$. Expanding, $q = n^2 + 3n + 2 = p + 2$.

Rewrite the product: $n(n+1)(n+2)(n+3) = p \cdot q = p(p+2)$. Expand $p(p+2) = p^2 + 2p$. Add 1 to both sides: $p^2 + 2p + 1$. Recognize the perfect square: $p^2 + 2p + 1 = (p+1)^2$.

Substitute back $p = n^2 + 3n$ to get $p + 1 = n^2 + 3n + 1$. Therefore $n(n+1)(n+2)(n+3) + 1 = (n^2 + 3n + 1)^2$. Since $n^2 + 3n + 1$ is an integer, $n(n+1)(n+2)(n+3) + 1$ is a perfect square.
\end{proof}

\begin{theorem}
For all integers $n$, $n(n+4)(n+8)(n+12) + 256$ is a perfect square.
\end{theorem}

\begin{proof}
Let $n$ be an arbitrary integer. Pair the outer terms: let $p = n \cdot (n+12)$. Expanding, $p = n^2 + 12n$. Pair the inner terms: let $q = (n+4) \cdot (n+8)$. Expanding, $q = n^2 + 12n + 32 = p + 32$.

Rewrite the product: $n(n+4)(n+8)(n+12) = p \cdot q = p(p+32)$. Expand $p(p+32) = p^2 + 32p$. Add 256 to both sides: $p^2 + 32p + 256$. Recognize the perfect square: $p^2 + 32p + 256 = (p+16)^2$.

Substitute back $p = n^2 + 12n$ to get $p + 16 = n^2 + 12n + 16$. Therefore $n(n+4)(n+8)(n+12) + 256 = (n^2 + 12n + 16)^2$. Since $n^2 + 12n + 16$ is an integer, $n(n+4)(n+8)(n+12) + 256$ is a perfect square.
\end{proof}

\subsection{Study Questions}
\label{questions}

\subsubsection{Task 1 Questions (Composite Numbers)}

\begin{enumerate}
    \item Describe the proof of Theorem 1 in your own words.
    \item Why does the proof of Theorem 1 show $x - 1 > 1$?
    \item Why does the theorem not hold for $x = 2$?
    \item Can the assumption in Theorem 1 ($x > 2$) be weakened while keeping the conclusion true?
    \item What are the similarities and differences between Theorem 1 and 4 in their proofs?
\end{enumerate}

\subsubsection{Task 2 Questions (Arithmetic Progression Squares)}

\begin{enumerate}
    \item Describe the proof of Theorem 1 in your own words.
    \item Why does the proof pair outer terms and inner terms separately?
    \item Why does this theorem hold for all integers?
    \item Can Theorem 1 be generalized by modifying the expression?
    \item What are the similarities and differences between Theorem 1 and 4 in their proofs?
\end{enumerate}

\end{document}